\documentclass[11pt]{article}

\usepackage{graphicx}
\usepackage{appendix}
\usepackage{amssymb}
\usepackage{dsfont}
\usepackage{amsmath}
\usepackage{amsthm}
\usepackage{cancel}
\usepackage{verbatim}
\usepackage{chngpage}
\usepackage{fullpage}

\newtheorem{definition}{Definition}
\newtheorem{theorem}{Theorem}
\newtheorem{proposition}[theorem]{Proposition}
\newtheorem{corollary}[theorem]{Corollary}

\newtheorem{lemma}[theorem]{Lemma}

\begin{document}
\title{Logical locality entails \\
frugal distributed computation over graphs}
\author{St\'{e}phane Grumbach  \hspace*{1cm} Zhilin Wu\\
INRIA-LIAMA\thanks{CASIA -- PO Box 2728 --
Beijing 100080 -- PR China -- Stephane.Grumbach@inria.fr  zlwu@liama.ia.ac.cn}\\
Chinese Academy of Sciences
}

\maketitle

\begin{abstract}
First-order logic is known to have limited expressive power over
finite structures. It enjoys in particular the locality property,
which states that first-order formulae cannot have a global view of
a structure. This limitation ensures on their low sequential
computational complexity. We show that the locality impacts as well
on their distributed computational complexity. We use first-order
formulae to describe the properties of finite connected graphs,
which are the topology of communication networks, on which the
first-order formulae are also evaluated. We show that over bounded
degree networks and planar networks, first-order properties can be
frugally evaluated, that is, with only a bounded number of messages,
of size logarithmic in the number of nodes, sent over each link.
Moreover, we show that the result carries over for the extension of
first-order logic with unary counting.

\end{abstract}

\vspace*{-1em}

\section{Introduction}


Logical formalisms have been widely used in many areas of computer
science to provide high levels of abstraction, thus offering
user-friendliness while increasing the ability to perform
verification.  In the field of databases, first-order logic
constitutes the basis of relational query languages, which allow to
write queries in a declarative manner, independently of the physical
implementation. In this paper, we propose to use logical formalisms
to express properties of the topology of communication networks,
that can be verified in a distributed fashion over the networks
themselves.

We focus on first-order logic over graphs. First-order logic has
been shown to have limited expressive power over finite structures.
In particular, it enjoys the locality property, which states that
all first-order formulae are local \cite{Gaifman82}, in the sense
that local areas of the graphs are sufficient to evaluate them.

First-order properties have been shown to be computable with very
low complexity in both sequential and parallel models of
computation. It was shown that first-order properties can be
evaluated in linear time over classes of bounded degree graphs
\cite{Seese95} and over classes of locally tree-decomposable
graphs\footnote{Locally tree-decomposable graphs generalize bounded
degree graphs, planar graphs, and graphs of bounded genus.}
\cite{FG01}. These results follow from the locality of the logic. It
was also shown that they can be evaluated in constant time over
Boolean circuits with unbounded fan-in (AC$^0$) \cite{Immerman89}.
These bounds lead us to be optimistic on the complexity of the
distributed evaluation of first-order properties.

We consider communication networks based on the message passing
model \cite{AttiyaW04}, where nodes exchange messages with their
neighbors. The properties to be evaluated concern the graph which
forms the topology of the network, and whose knowledge is
distributed over the nodes, who are only aware of their $1$-hop
neighbors. We thus focus on connected graphs.

In distributed computing, the ability to solve problems locally has
attracted a strong interest since the seminal paper of Linial
\cite{Linial92}.  The ability to solve global problems in
distributed systems, while performing as much as possible local
computations, is of great interest in particular to ensure
scalability. Moreover relying as much as possible on local
information improves fault-tolerance. Finally, restricting the
computation to local areas allows to optimize time and communication
complexity.

Naor and Stockmeyer \cite{NaorS95} showed that there were
non-trivial locally checkable labelings that are locally computable,
while on the other hand lower-bounds have been exhibited, thus
resulting in non-local computability results
\cite{KuhnMW04,KuhnMW06}.

Different notions of local computation have been considered. The
most widely accepted restricts the time of the computation to be
constant, that is independent of the size of the network
\cite{NaorS95}, while allowing messages of size $O(\log n)$, where
$n$ is the size of the network. This condition is rather stringent.
Naor and Stockmeyer \cite{NaorS95} show their result for a
restricted class of graphs (eg bounded odd degree). Godard et al.
used graph relabeling systems as the distributed computational
model, defined local computations as graph relabeling systems with
locally-generated local relabeling rules, and characterized the
classes of graphs that are locally computable \cite{GodardMM04}.

Our initial motivation is to understand the impact of the logical
locality on the distributed computation, and its relationship with
local distributed computation. It is easy to verify though that
there are simple properties (expressible in first-order logic) that
cannot be computed locally. Consider for instance the property
``There exist at least two distinct triangles'', which requires
non-local communication to check the distinctness of the two
triangles which may be far away from each other. Nevertheless,
first-order properties do admit simple distributed computations.

We thus introduce frugal distributed computations. A distributed
algorithm is \emph{frugal} if during its computation only a bounded
number of messages of size $O(\log n)$ are sent over each link. If
we restrict our attention to bounded degree networks, this implies
that each node is only receiving a bounded number of messages.
Frugal computations resemble local computations over bounded degree
networks, since the nodes are receiving only a bounded number of
messages, although these messages can come from remote nodes through
multi-hop paths.

We prove that first-order properties can be frugally evaluated over
bounded degree networks and planar networks
(Theorem~\ref{thm:FO-bounded-degree} and
Theorem~\ref{thm:FO-planar}). The proofs are obtained by
transforming the centralized linear time evaluation algorithms
\cite{Seese95,FG01} into distributed ones satisfying the restriction
that only a bounded number of messages are sent over each link.
Moreover, we show that the results carry over to the extension of
first-order logic with unary counting. While the transformation of
the centralized linear time algorithm is simple for first-order
properties over bounded degree networks, it is quite intricate for
first-order properties over planar networks. The most intricate part
is the distributed construction of an ordered tree decomposition for
some subgraphs of the planar network, inspired by the distributed
algorithm to construct an ordered tree decomposition for planar
networks with bounded diameter in \cite{GW09}.

Intuitively, since in the centralized linear time computation each
object is involved only a bounded number of times, in the
distributed computation, a bounded number of messages sent over each
link could be sufficient to evaluate first-order properties. So it
might seem trivial to design frugal distributed algorithms for
first-order properties over bounded degree networks and planar
networks. Nevertheless, this is not the case, because in the
centralized computation, after visiting one object, any other object
can be visited, but in the distributed computation, only the
\emph{adjacent} objects (nodes, links) can be visited.

The paper is organized as follows. In the next section, we recall
classical graph theory concepts, as well as  Gaifman's locality
theorem. In Section~\ref{sec-dist-FO}, we consider the distributed
evaluation of first-order properties over respectively bounded
degree and planar networks. Finally, in Section~\ref{sec-beyond-FO},
we consider the distributed evaluation of first-order logic with
unary counting. Proofs can be found in the appendix.

\vspace*{-1em}

\section{Graphs, first-order logic and locality}\label{sec-FO}

In this paper, our interest is focused to a restricted class of
structures, namely finite graphs. Let $G=(V,E)$, be a finite graph.
We use the following notations. If $v \in V$, then $deg(v)$ denotes
the {\it degree} of $v$. For two nodes $u,v \in V$, the {\it
distance} between $u$ and $v$, denoted $dist_G(u,v)$, is the length
of the shortest path between $u$ and $v$. For $k \in \mathds{N}$,
the {\it $k$-neighborhood} of a node $v$, denoted $N_k(v)$, is
defined as $\{w \in V | dist_G(v,w) \le k\}$. If $\bar{v}=v_1...v_p$
is a collection of nodes in $V$, then the $k$-neighborhood of
$\bar{v}$, denoted $N_k(\bar{v})$, is defined by $\bigcup_{1 \le i
\le p} N_k(v_i)$.
For $X \subseteq V$, let $\langle X \rangle^G$ denote the subgraph
induced by $X$.

Let $G=(V,E)$ be a connected graph, a \emph{tree decomposition} of
$G$ is a rooted labeled tree $\mathcal{T}=(T,F, r, B)$, where $T$ is
the set of vertices of the tree, $F \subseteq T \times T$ is the
child-parent relation of the tree, $r \in T$ is the root of the
tree, and $B$ is a labeling function from $T$ to $2^V$, mapping
vertices $t$ of $T$ to sets $B(t)\subseteq V$, called \emph{bags},
such that
\vspace*{-2mm}
\begin{enumerate}
\item For each edge $(v,w) \in E$, there is a $t \in T$, such that $\{v,w\} \subseteq B(t)$.
\item For each $v \in V$, $B^{-1}(v)=\{t \in T | v \in B(t)\}$ is connected in
$T$.
\end{enumerate}
\vspace*{-2mm}
The {\it width} of $\mathcal{T}$, $width(\mathcal{T})$, is defined
as $\max\{|B(t)|-1 | t \in T\}$. The tree-width of $G$, denoted
$tw(G)$, is the minimum width over all tree decompositions of $G$.
An \emph{ordered tree decomposition} of width $k$ of a graph $G$ is
a rooted labeled tree $\mathcal{T}=(T,F,r,L)$ such that:
\vspace*{-2mm}
\begin{itemize}
\item $(T,F,r)$ is defined as above,

\item $L$ assigns each vertex $t \in T$ to a $(k+1)$-tuple
$\overline{b^t}=(b^t_1,\cdots,b^t_{k+1})$ of vertices of $G$ (note
that in the tuple $\overline{b^t}$, vertices of $G$ may occur
repeatedly),

\item If $L^\prime(t):=\{b^t_j| L(t)= (b^t_1,\cdots,b^t_{k+1}), 1 \le j \le k+1\}$, then
$(T,F,r,L^\prime)$ is a tree decomposition.
\end{itemize}
\vspace*{-2mm}
The \emph{rank} of an (ordered) tree decomposition is the rank of
the rooted tree, i.e. the maximal number of children of its
vertices.

We consider first-order logic (FO) over the signature $E$, where $E$
is a binary relation symbol. The syntax and semantics of first-order
formulae are defined as usual \cite{EbbinghausFlum99}. The
\emph{quantifier rank} of a formula $\varphi$ is the maximal number
of nestings of existential and universal quantifiers in $\varphi$.

A \emph{graph property} is a class of graphs closed under
isomorphisms. Let $\varphi$ be a first-order sentence, the graph
property defined by $\varphi$, denoted $\mathcal{P}_\varphi$, is the
class of graphs satisfying $\varphi$.


The distance between nodes can be defined by first-order formulae
$dist(x,y) \le k$ stating that the distance between $x$ and $y$ is
no larger than $k$, and $dist(x,y) > k$ is an abbreviation of $\neg
dist(x,y) \le k$. In addition, let $\bar{x}=x_1...x_p$ be a list of
variables, then $dist(\bar{x},y) \le k$ is used to denote
$\mathop{\vee} \limits_{1 \le i \le p} dist(x_i,y) \le k$.

Let $\varphi$ be a first-order formula, $k \in \mathds{N}$, and
$\bar{x}$ be a list of variables not occurring in $\varphi$, then
the formula bounding the quantifiers of $\varphi$ to the
$k$-neighborhood of $\bar{x}$, denoted $\varphi^{(k)}(\bar{x})$, can
be defined easily in first-order logic by using formulae
$dist(\bar{x},y) \le k$. For instance, if $\varphi:= \exists y
\psi(y)$, then
\[\varphi^{(k)}(\bar{x}) := \exists y \left(dist(\bar{x},y) \le k \wedge
\left(\psi(y)\right)^{(k)}(\bar{x})\right).\]

We can now recall the notion of logical  locality introduced by
Gaifman  \cite{Gaifman82,EbbinghausFlum99}.
\vspace*{-.5em}
\begin{theorem}\label{the-gaifman}
\cite{Gaifman82} Let $\varphi$ be a first-order formula with free
variables $u_1,...,u_p$, then $\varphi$ can be written in {\it
Gaifman Normal Form}, that is into a Boolean combination of (i)
sentences of the form:
\begin{equation}\label{eqn:gaifman-theorem}
\exists x_1 ... \exists x_s \left(\bigwedge \limits_{1 \le i<j \le s} d(x_i,x_j)
> 2r \wedge \bigwedge \limits_{i} \psi^{(r)}(x_i)\right)
\end{equation}
and (ii) formulae of the form $\psi^{(t)}(\overline{y})$, where
$\overline{y} = y_1...y_q$ such that $y_i \in \{u_1,...,u_p\}$ for
all $1 \le i \le q$, $r \le 7^{k-1}$, $s \le p + k$, $t \le
\left(7^k-1\right)/2$ ($k$ is the quantifier rank of
$\varphi$)\footnote{The bound on $r$ has been improved to $4^k-1$ in
\cite{KeislerL04}}. \\
Moreover, if $\varphi$ is a sentence, then the
Boolean combination contains only sentences of the form
(\ref{eqn:gaifman-theorem}).
\end{theorem}

The locality of first-order logic is a powerful tool to demonstrate
non-definability results \cite{Libkin97}. It can be used in
particular to prove that counting properties, such as the parity of
the number of vertices, or recursive properties, such as the
connectivity of a graph, are not first-order.

\vspace*{-1em}
\section{Distributed evaluation of FO}\label{sec-dist-FO}
We consider a message passing model of distributed computation
\cite{AttiyaW04}, based on a communication network whose topology is
given by a graph $G=(V,E)$ of diameter $\Delta$, where  $E$ denotes
the set of bidirectional \emph{communication links} between nodes.
From now on, we restrict our attention to  \emph{finite connected
graphs}.

We assume that the distributed system is asynchronous and has no
failure. The nodes have a unique \emph {identifier} taken from
$1,2,\cdots,n$, where $n$ is the number of nodes. Each node has
distinct local ports for distinct links incident to it. The nodes
have {\it states}, including final accepting or rejecting states.

For simplicity, we assume that there is only one query fired in the
network by a {\it requesting node}. We assume also that a {\it
breadth-first-search (BFS) tree} rooted on the requesting node has
been pre-computed in the network, such that each node stores locally
the identifier of its parent in the BFS-tree, and the states of the
ports with respect to the BFS-tree, which are either ``parent'' or
``child'', denoting the ports corresponding to the tree edges, or
``horizon'', ``upward'', ``downward'', denoting the ports
corresponding to the non-tree edges to some node with the same,
smaller, or larger depth in the BFS-tree. The computation
terminates, when the requesting node reaches a final state.

Let $\mathcal{C}$ be a class of graphs. A distributed algorithm is
said to be \emph{frugal} over $\mathcal{C}$ if there is a $k \in
\mathds{N}$ such that for any network $G \in \mathcal{C}$ of $n$
nodes and any requesting node in $G$, the distributed computation
terminates, with only at most $k$ messages of size $O(\log n)$ sent
over each link. If we restrict our attention to bounded degree
networks, frugal distributed algorithms implies that each node only
receives a bounded number of messages. Frugal computations resemble
local computations over bounded degree networks, since the nodes
receive only a bounded number of messages, although these messages
can come from remote nodes through multi-hop paths.

Let $\mathcal{C}$ be a class of graphs, and $\varphi$ an FO
sentence, we say that $\varphi$ can be distributively evaluated over
$\mathcal{C}$ if there exists a distributed algorithm such that for
any network $G \in \mathcal{C}$ and any requesting node in $G$, the
computation of the distributed algorithm on $G$ terminates with the
requesting node in the accepting state if and only if $G \models
\varphi$. Moreover, if there is a frugal distributed algorithm to do
this, then we say that $\varphi$ can be frugally evaluated over
$\mathcal{C}$.


For centralized computations, it has been shown that Gaifman's
locality of FO entails linear time evaluation of FO properties over
classes of bounded degree graphs and classes of locally
tree-decomposable graphs \cite{Seese95,FG01}. In the following, we
show that it is possible to design frugal distributed evaluation
algorithms for FO properties over bounded degree and planar
networks, by carefully transforming the centralized linear time
evaluation algorithms into distributed ones with computations on
each node well balanced.

\vspace*{-.5em}

\subsection{Bounded degree networks}\label{sec-FO-bounded-degree}

We first consider the evaluation of FO properties over bounded
degree networks. We assume that each node stores the degree bound
$k$ locally.

\vspace*{-.5em}

\begin{theorem}\label{thm:FO-bounded-degree}
FO properties can be frugally evaluated over bounded degree
networks.
\end{theorem}

\vspace*{-.5em}

Theorem~\ref{thm:FO-bounded-degree} can be shown by using Hanf's
technique \cite{FSV95}, in a way similar to the proof of Seese's
seminal result \cite{Seese95}.

Let $r \in \mathds{N}$, $G=(V,E)$, and $v \in V$, then the
\emph{$r$-type} of $v$ in $G$ is the isomorphism type of $\left(
\langle N_r(v) \rangle^G,v \right)$. Let $r, m \in \mathds{N}$,
$G_1$ and $G_2$ be two graphs, then $G_1$ and $G_2$ are said to be
$(r,m)$-equivalent if and only if for every $r$-type $\tau$, either
$G_1$ and $G_2$ have the same number of vertices with $r$-type
$\tau$ or else both have at least $m$ vertices with $r$-type $\tau$.
$G_1$ and $G_2$ are said to be $k$-equivalent, denoted $G_1
\equiv_{k} G_2$, if $G_1$ and $G_2$ satisfy the same FO sentences of
quantifier rank at most $k$. It has been shown that:

\vspace*{-.5em}

\begin{theorem}\label{thm:hanf-theorem}
\cite{FSV95} Let $k,d \in \mathds{N}$. There exist $r,m \in
\mathds{N}$ such that $r$ (resp. $m$) depends on $k$ (resp. both $k$
and $d$), and for any graphs $G_1$ and $G_2$ with maximal degree no
more than $d$, if $G_1$ and $G_2$ are $(r,m)$-equivalent, then $G_1
\equiv_k G_2$.
\end{theorem}

\vspace*{-.5em}

Let us now sketch the proof of Theorem~\ref{thm:FO-bounded-degree},
which relies on a distributed algorithm consisting of three phases.
Suppose the requesting node requests the evaluation of some FO
sentence with quantifier rank $k$. Let $r,m$ be the natural numbers
depending on $k,d$ specified in Theorem~\ref{thm:hanf-theorem}.

\vspace*{-.5em}

\begin{description}
\item[Phase I] The requesting node broadcasts messages along the BFS-tree to ask each node to collect the topology information in its $r$-neighborhood;
\vspace*{-.5em}
\item[Phase II] Each node collects the topology information in its
$r$-neighborhood;
\vspace*{-.5em}
\item[Phase III] The $r$-types of the nodes in the network are aggregated
through the BFS-tree to the requesting node up to the threshold $m$
for each $r$-type. Finally the requesting node decides whether the
network satisfies the FO sentence or not by using the information
about the $r$-types.
\end{description}

\vspace*{-.5em}

It is easy to see that only a bounded number of messages are sent
over each link in Phase I and II. Since the total number of distinct
$r$-types with degree bound $d$ depends only upon $r$ and $d$ and
each $r$-type is only counted up to a threshold $m$, it turns out
that over each link, only a bounded number of messages are sent in
Phase III as well. So the above distributed evaluation algorithm is
frugal over bounded degree networks.

\vspace*{-.5em}

\subsection{Planar networks}\label{sec-FO-planar}

We now consider the distributed evaluation of FO properties over
planar networks.

A \emph{combinatorial embedding} of a planar graph $G=(V,E)$ is an
assignment of a cyclic ordering of the set of incident edges to each
vertex $v$ such that two edges $(u,v)$ and $(v,w)$ are in the same
face iff $(v,w)$ is immediately before $(v,u)$ in the cyclic
ordering of $v$. Combinatorial embeddings, which  encode the
information about boundaries of the faces in usual embeddings of
planar graphs into the planes, are useful for computing on planar
graphs. Given a combinatorial embedding, the boundaries of all the
faces can be discovered by traversing the edges according to the
above condition.

We assume in this subsection that a combinatorial embedding of the
planar network is distributively stored in the network, i.e. a
cyclic ordering of the set of the incident links is stored in each
node of the network.
\vspace*{-.5em}

\begin{theorem}\label{thm:FO-planar}
FO properties can be frugally evaluated over planar networks.
\end{theorem}

\vspace*{-.5em}

For the proof of Theorem~\ref{thm:FO-planar}, we first recall the
centralized linear time algorithm to evaluate FO properties over
planar graphs in \cite{FG01}\footnote{In fact, in \cite{FG01}, it
was shown that FO is linear-time computable over classes of locally
tree-decomposable graphs.}.

Let $G=(V,E)$ be a planar graph and $\varphi$ be an FO sentence.
From Theorem~\ref{the-gaifman}, we know that $\varphi$ can be
written into Boolean combinations of sentences of the form
(\ref{eqn:gaifman-theorem}),

\vspace*{-.5em}

\begin{equation*}
\exists x_1 ... \exists x_s \left(\bigwedge \limits_{1 \le i<j \le
s} d(x_i,x_j)
> 2r \wedge \bigwedge \limits_{i} \psi^{(r)}(x_i)\right).
\end{equation*}

\vspace*{-.5em}

It is sufficient to show that sentences of the form
(\ref{eqn:gaifman-theorem}) are linear-time computable over $G$. The
centralized algorithm to evaluate FO sentences of the form
(\ref{eqn:gaifman-theorem}) over planar graphs consists of the
following four phases:

\vspace*{-.5em}

\begin{enumerate}
\item Select some $v_0 \in V$, let $\mathcal{H}=\{G[i,i+2r] | i \ge 0\}$, where $G[i,j]=\{v \in V |
i \le dist_G(v_0,v) \le j\}$;

\item For each $H \in \mathcal{H}$, compute $K_r(H)$, where $K_r(H):=\{v \in H | N_r(v) \subseteq H\}$;

\item For each $H \in \mathcal{H}$, compute $P_H:=\{ v \in K_r(H) | \langle H \rangle^G \models
\psi^{(r)}(v)\}$;

\item Let $P:=\cup_H P_H$, determine whether there are $s$ distinct nodes in $P$ such
that their pairwise distance is greater than $2r$.
\end{enumerate}

\vspace*{-.5em}

In the computation of the 3rd and 4th phase above, an
automata-theoretical technique to evaluate Monadic-Second-Order
(MSO) formulae in linear time over classes of graphs with bounded
tree-width \cite{Courcelle90,FlumG06,FFG02} is used. In the
following, we recall this centralized evaluation algorithm.

MSO is obtained by adding set variables and set quantifiers into FO,
such as $\exists X \varphi(X)$ (where $X$ is a set variable). MSO
has been widely studied in the context of graphs for its expressive
power. For instance, $3$-colorability, transitive closure or
connectivity can be defined in MSO \cite{Courcelle08}.

The centralized linear time evaluation of MSO formulae over classes
of bounded tree-width graphs goes as follows. First an ordered tree
decomposition $\mathcal{T}$ of the given graph is constructed. Then
from the given MSO formula, a tree automaton $\mathcal{A}$ is
obtained. Afterwards, $\mathcal{T}$ is transformed into a labeled
tree $\mathcal{T}^\prime$, finally $\mathcal{A}$ is ran over
$\mathcal{T}^\prime$ (maybe several times for formulae containing
free variables) to get the evaluation result.

\medskip

In the rest of this section, we design a frugal distributed
algorithm to evaluate FO sentences over planar networks by adapting
the above centralized algorithm. The main difficulty is to
distribute the computation among the nodes such that only a bounded
number of messages are sent over each link during the computation.

\vspace*{-.5em}

\begin{description}
\item[Phase I] The requesting node broadcasts the FO sentence of the form (\ref{eqn:gaifman-theorem}) to all the
nodes in the network through the BFS tree;
\vspace*{-.5em}

\item[Phase II] For each $v \in V$, compute $C(v):=\{i \ge 0 | v \in G[i,i+2r]\}$;
\vspace*{-.5em}

\item[Phase III] For each $v \in V$, compute $D(v):=\{i \ge 0 | N_r(v) \subseteq G[i,i+2r]\}$;
\vspace*{-.5em}

\item[Phase IV] For each $i \ge 0$, compute $P_i:=\{ v \in V | i \in D(v), \langle G[i,i+2r] \rangle^G \models \psi^{(r)}(v)\}$;
\vspace*{-.5em}

\item[Phase V] Let $P:=\bigcup_i P_i$, determine whether there are $s$ distinct nodes labeled by $P$ such that
their pairwise distance is greater than $2r$.
\end{description}

\vspace*{-.5em}

Phase I is trivial. Phase II is easy. In the following, we
illustrate the computation of Phase III, IV, and V one by one.

We first introduce a lemma for the computation of Phase III.

For $W \subseteq V$, let $K_i(W):=\{v \in W | N_i(v) \subseteq W\}$.
Let $D_i(v):=\{j \ge 0 | v \in K_i(G[j,j+2r])\}$.

\vspace*{-.5em}

\begin{lemma}\label{lem-kernel}
For each $v \in V$ and $i > 0$, $D_i(v)=C(v) \cap \bigcap
\limits_{w: (v,w) \in E} D_{i-1}(w)$.
\end{lemma}

\vspace*{-.5em}

With Lemma~\ref{lem-kernel}, $D(v)=D_r(v)$ can be computed in an
inductive way to finish Phase III: Each node $v$ obtains the
information $D_{i-1}(w)$ from all its neighbors $w$, and performs
the in-node computation to compute $D_i(v)$.

\smallskip

Now we consider Phase IV.

Because $\psi^{(r)}(x)$ is a local formula, $\psi^{(r)}(x)$ can be
evaluated separately over each connected component of $G[i,i+2r]$
and the results are stored distributively.

Let $C_i$ be a connected component of $G[i,i+2r]$, and
$w^i_1,\cdots,w^i_l$ be all the nodes contained in $C_i$ with
distance $i$ from the requesting node. Now we consider the
evaluation of $\psi^{(r)}(x)$ over $C_i$.

Let $C^\prime_i$ be the graph obtained from $C_i$ by including all
ancestors of $w^i_1,\cdots,w^i_l$ in the BFS-tree, and $C^\ast_i$ be
the graph obtained from $C^\prime_i$ by contracting all the
ancestors of $w^i_1,\cdots,w^i_l$ into one vertex, i.e. $C^\ast_i$
has one more vertex, called the virtual vertex, than $C_i$, and this
vertex is connected to $w^i_1,\cdots,w^i_l$. It is easy to see that
$C^\ast_i$ is a planar graph with a BFS-tree rooted on $v^\ast$ and
of depth at most $2r+1$. So $C^\ast_i$ is a planar graph with
bounded diameter.

An ordered tree decomposition for planar networks with bounded
diameter can be distributively constructed with only a bounded
number of messages sent over each link as follows \cite{GW09}:
\vspace*{-.5em}

\begin{itemize}
\item Do a depth-first-search to decompose the network into blocks,
i.e. biconnected components;
\vspace*{-.5em}

\item Construct an ordered tree decomposition for each
nontrivial block: Traverse every face of the block according to the
cyclic ordering at each node, triangulate all those faces, and
connect the triangles into a tree decomposition by utilizing the
pre-computed BFS tree;
\vspace*{-.5em}

\item Finally the tree decompositions for the blocks are connected
together into a complete tree decomposition for the whole network.
\end{itemize}
\vspace*{-.5em}

By using the distributed algorithm for the tree decomposition of
planar networks with bounded diameter, we can construct
distributively an ordered tree decomposition for $C^\ast_i$, while
having the virtual vertex in our mind, and get an ordered tree
decomposition for $C_i$.

With the ordered tree decomposition for $C_i$, we can evaluate
$\psi^{(r)}(x)$ over $C_i$ by using the automata-theoretical
technique, and store the result distributively in the network (each
node stores a Boolean value indicating whether it belongs to the
result or not).

A distributed post-order traversal over the BFS tree can be done to
find out all the connected components of $G[i,i+2r]$'s and construct
the tree decompositions for these connected components one by one.

Finally we consider Phase V.

Label nodes in $\bigcup_i P_i$ with $P$.

Then consider the evaluation of FO sentence $\varphi^\prime$ over
the vocabulary $\{E,P\}$,
\vspace*{-.5em}

\begin{equation*}
\exists x_1 ... \exists x_s \left(\bigwedge \limits_{1 \le i<j \le
s} d(x_i,x_j)
> 2r \wedge \bigwedge \limits_{i} P(x_i)\right).
\end{equation*}
\vspace*{-.5em}

Starting from some node $w_1$ with label $P$, mark the vertices in
$N_{2r}(w_1)$ as $Q$, then select some node $w_2$ outside $Q$, and
mark those nodes in $N_{2r}(w_2)$ by $Q$ again, continue like this,
until $w_l$ such that either $l = s$ or all the nodes with label $P$
have already been labeled by $Q$.

If $l < s$, then label the nodes in $\bigcup \limits_{1 \le i \le l}
N_{4r}(v_i)$ as $I$. Each connected component of $\langle I
\rangle^G$ has diameter no more than $4lr < 4sr$. We can construct
distributively a tree decomposition for each connected component of
$\langle I \rangle^G$, and connect these tree decompositions
together to get a complete tree-decomposition of $\langle I
\rangle^G$, then evaluate the sentence $\varphi^\prime$ by using
this complete tree decomposition.

The details of the frugal distributed evaluation algorithm can be
found in the appendix.
\vspace*{-1em}

\section{Beyond FO properties}\label{sec-beyond-FO}
\vspace*{-.5em}
We have shown that FO properties can be frugally evaluated over
respectively bounded degree and planar networks. In this section, we
extend these results to FO unary queries and some counting extension
of FO.

From Theorem~\ref{the-gaifman}, FO formula $\varphi(x)$ containing
exactly one free variable $x$ can be written into the Boolean
combinations of sentences of the form (1) and the local formulae
$\psi^{(t)}(x)$. Then it is not hard to prove the following result.
\vspace*{-.5em}
\begin{theorem}\label{thm:FO-Unary}
FO formulae $\varphi(x)$ with exactly one free variable $x$ can be
frugally evaluated over respectively bounded degree and planar
networks, with the results distributively stored on the nodes of the
network.
\end{theorem}
\vspace*{-.5em}
%

Counting is one of the ability that is lacking to first-order logic,
and has been added in commercial relational query languages (e.g.
SQL). Its expressive power has been widely studied
\cite{GradelO92,GrumbachT95,Otto96} in the literature. Libkin
\cite{Libkin97} proved that first-order logic with counting still
enjoys Gaifman locality property. We prove that
Theorem~\ref{thm:FO-bounded-degree} and Theorem~\ref{thm:FO-planar}
carry over as well for first-order logic with unary counting.

Let FO($\#$) be the extension of first-order logic with unary
counting. FO($\#$) is a two-sorted logic, the first sort ranges over
the set of nodes $V$, while the second sort ranges over the natural
numbers $\mathds{N}$. The terms of the second sort are defined by:
$t:= \# x.\varphi(x) \  | \ t_1 + t_2 \ | \ t_1 \times t_2$, where
$\varphi$ is a formula over the first sort with one free variable
$x$. Second sort terms of the form $\# x. \varphi(x)$ are called
\emph{basic} second sort terms.

The atoms of FO($\#$) extend standard FO atoms with the following
two unary counting atoms: $t_1 = t_2 \ | \ t_1 < t_2,$ where $ t_1,
t_2$ are second sort terms. Let $t$ be a second sort term of
FO($\#$), $G=(V,E)$ be a graph, then the interpretation of $t$ in
$G$, denoted $t^G$, is defined as follows:
\vspace*{-.5em}
\begin{itemize}
\item $(\# x.\varphi(x))^G$ is the cardinality of $\{v \in V | G\models\varphi(v) \}$;
\vspace*{-.5em}
\item $\left(t_1 + t_2\right)^G$ is the sum of $t_1^G$ and $t_2^G$;
\vspace*{-.5em}
\item $\left(t_1 \times t_2\right)^G$ is the product of $t_1^G$ and $t_2^G$.
\end{itemize}
\vspace*{-.5em}
The interpretation of FO($\#$) formulae is defined in a standard
way.

%
%
\vspace*{-.5em}
\begin{theorem}\label{thm:FO-UCNT-frugal}
FO($\#$) properties can be frugally evaluated over respectively
bounded degree and planar networks.
\end{theorem}
\vspace*{-.5em}
The proof of the theorem relies on a normal form of FO($\#$)
formulae. A sketch can be found in the appendix.

\vspace*{-1em}
\section{Conclusion}\label{sec-conclusion}

The logical locality has been shown to entail efficient computation
of first-order logic over several classes of structures. We show
that if the logical formulae are used to express properties of the
graphs, which constitute the topology of communication networks,
then these formulae can be evaluated very efficiently over these
networks. Their distributed computation, although not local
\cite{Linial92,NaorS95,Peleg00}, can be done {\it frugally}, that is
with a bounded number of messages of logarithmic size exchanged over
each link. The frugal computation, introduced in this paper,
generalizes local computation and offers a large spectrum of
applications. We proved that first-order properties can be evaluated
frugally over respectively bounded degree and planar networks.
Moreover the results carry over to the extension of first-order
logic with unary counting. The distributed time used in the frugal
evaluation of FO properties over bounded degree networks is
$O(\Delta)$, while that over planar networks is $O(n)$.

We assumed that some pre-computations had been done on the networks.
If no BFS-tree has been pre-computed, the construction of a BFS-tree
can be done in $O(\Delta)$ time and with $O(\Delta)$ messages sent
over each link \cite{BDLP08}.

Beyond its interest for logical properties, the frugality of
distributed algorithms, which ensures an extremely good scalability
of their computation, raises fundamental questions, such as deciding
what can be frugally computed. Can a Hamiltonian path for instance
be computed frugally?



\newpage

\normalsize

\begin{appendix}


\section{Distributed Evaluation of FO over planar networks: Phase II}\label{appendix-FO-planar-phase-ii}

The purpose of Phase II is to compute $C(v):=\{i \ge 0 | v \in
G[i,i+2r]\}$ for each $v \in V$.

A pre-computed breadth-first-search (BFS) tree rooted on the
requesting node is distributively stored in the network such that
each node $v$ stores the identifier of its parent in the BFS-tree
($parent(v)$), and the states of the ports with respect to the
BFS-tree ($state(l)$ for each port $l$), which are either
``parent'', or ``child'', or ``horizon'', or ``downward'', or
``upward''. Moreover, we suppose that each node $v$ stores in
$depth(v)$ its depth in the BFS tree, i.e. the distance between $v$
and the requesting node.

The distributed algorithm is presented by describing the message
processing at each node $v$.

\begin{table}[ht]
  \centering
  \small
  \begin{tabular*}{\textwidth}{l@{\extracolsep{\fill}}}
    \hline
    Initialization\\
    \hline
    The requesting node sets $treeDepth:=0$.\\
    The requesting node sends message TREEDEPTH over all its ports with state ``child''.\\
    \hline
    \\
    \hline
    Message TREEDEPTH over port $l$\\
    \hline
    $treeDepth:=depth(v)$.\\
    \textbf{if} $v$ is not a leaf \textbf{then}\\
    \hspace*{1em} $v$ sends message TREEDEPTH over all ports with state ``child''.\\
    \textbf{else} $v$ sends message ACKTREEDEPTH($treeDepth$) over the port $l^\prime$ with state ``parent''.\\
    \textbf{end if}\\
    \hline
    \\
    \hline
    Message ACKTREEDEPTH($sd$) over port $l$\\
    \hline
    $treeDepth:=\max\{treeDepth,sd\}$.\\
    \textbf{if} $v$ has received ACKTREEDEPTH messages over all its ports with state ``child'' \textbf{then} \\
    \hspace*{1em} \textbf{if} $v$ is the requesting node \textbf{then}\\
    \hspace*{1em}\hspace*{1em} $v$ sends message STARTCOVER($treeDepth$) over all ports with state ``child''.\\
    \hspace*{1em} \textbf{else} $v$ sends message ACKTREEDEPTH($treeDepth$) over the port $l^\prime$ with state ``parent''.\\
    \hspace*{1em} \textbf{end if}\\
    \textbf{end if}\\
    \hline
    \\
    \hline
    Message STARTCOVER($td$) over port $l$\\
    \hline
    $treeDepth:=td$.\\
    \textbf{if} $treeDepth \le 2r$ \textbf{then}\\
    \hspace*{1em} $C(v):=\{0\}$.\\
    \textbf{else}\\
    \hspace*{1em} $C(v):=\left\{i \in \mathds{N} | \max\{depth(v)-2r,0\} \le i \le
    \min\{depth(v),treeDepth-2r\}\right\}$.\\
    \textbf{end if}\\
    \textbf{if} $v$ is not a leaf \textbf{then}\\
    \hspace*{1em} $v$ sends message STARTCOVER($treeDepth$) over all ports with state ``child''.\\
    \textbf{else} $v$ sends message ACKCOVER over the port $l^\prime$ with state ``parent''.\\
    \textbf{end if}\\
    \hline
    \\
    \hline
    Message ACKCOVER over port $l$\\
    \hline
    \textbf{if} $v$ has received message ACKCOVER over all its ports with state ``child'' \textbf{then}\\
    \hspace*{1em} \textbf{if} $v$ is not the requesting node \textbf{then}\\
    \hspace*{1em}\hspace*{1em} $v$ sends message ACKCOVER over the port $l^\prime$ with state ``parent''.\\
    \hspace*{1em} \textbf{end if}\\
    \textbf{end if}\\
    \hline
  \end{tabular*}
\end{table}

\clearpage

\section{Distributed Evaluation of FO over planar networks: Phase III}\label{appendix-FO-planar-phase-iii}

The purpose of Phase III is to compute $D(v):=\{i \ge 0 | N_r(v)
\subseteq G[i,i+2r]\}$ for each $v \in V$.

When the requesting node receives message ACKCOVER from all its
children, it knows that the computation of Phase II is over. Then it
can starts the computation of Phase III.

We first introduce a lemma.

For $W\subseteq V$, let $K_i(W):=\{v \in W | N_i(V) \subseteq W\}$.
Let $D_i(v):=\{j \ge 0 | v \in K_i[G[j,j+2r]]\}$.

\medskip

\noindent \textbf{Lemma \ref{lem-kernel}}. \textit{For each $v \in
V$ and $i
> 0$, $D_i(v)=C(v) \cap \bigcap \limits_{w: (v,w) \in E}
D_{i-1}(w)$.}

\begin{proof}
\begin{eqnarray*}
  j \in D_i(v) & \Leftrightarrow & v \in K_i(G[j,j+2r]) \Leftrightarrow N_i(v) \subseteq G[j,j+2r]\\
  & \Leftrightarrow & v \in G[j,j+2r] \ and \  \forall w \left( (v,w) \in E \rightarrow N_{i-1}(w) \subseteq G[j,j+2r]\right)\\
  & \Leftrightarrow & j \in C(v) \ and\  \forall w \left((v,w) \in E \rightarrow w \in K_{i-1}(G[j,j+2r])\right)\\
  & \Leftrightarrow & j \in C(v) \ and\  \forall w \left((v,w) \in E \rightarrow j \in D_{i-1}(w)\right)\\
  & \Leftrightarrow & j \in C(v) \cap \bigcap \limits_{w: (v,w) \in E} D_{i-1}(w)\\
\end{eqnarray*}
\end{proof}

With Lemma~\ref{lem-kernel}, $D(v)$'s can be computed in an
inductive way: Each node $v$ obtains the information $D_{i-1}(w)$
from all its the neighbors $w$, and does the in-node computation.

The distributed algorithm is given by describing the message
processing at each node $v$.
\\

\begin{table}[ht]
  \centering
  \small
  \begin{tabular*}{\textwidth}{l@{\extracolsep{\fill}}}
    \hline
    Initialization \\
    \hline
    The requesting node sends message INIT over all ports with state ``child''.\\
    \hline
    \\
    \hline
    Message INIT over port $l$\\
    \hline
    $idx(v):=1$, $D(v):= C(v)$. \\
    \textsf{\% $idx(v)$ is the index $i$ such that $D_i(v)$ is to be computed next.}\\
    $neigborKernel(v):=\emptyset$.\\
    $bKernelOver(v):=false$.\\
    \textbf{if} $v$ is not a leaf \textbf{then}\\
    \hspace*{1em} $v$ sends message INIT over all ports with state ``child''.\\
    \textbf{else} \\
    \hspace*{1em} $v$ sends message ACKINIT over the port $l^\prime$ with state ``parent''.\\
    \textbf{end if}\\
    \hline
    \\
    \hline
    Message ACKINIT over port $l$\\
    \hline
    \textbf{if} $v$ has received ACKINIT messages over all ports with state ``child'' \textbf{then}\\
    \hspace*{1em} \textbf{if} $v$ is not the requesting node \textbf{then}\\
    \hspace*{1em}\hspace*{1em} $v$ sends message ACKINIT over the port $l^\prime$ with state ``parent''.\\
    \hspace*{1em} \textbf{else} \\
    \hspace*{1em}\hspace*{1em} $v$ sends message STARTKERNEL over all ports with state ``child''.\\
    \hspace*{1em} \textbf{end if}\\
    \textbf{end if}\\
    \hline
  \end{tabular*}
\end{table}

\bigskip

\begin{table}[ht]
  \centering
  \small
  \begin{tabular*}{\textwidth}{l@{\extracolsep{\fill}}}
    \hline
    Message STARTKERNEL over port $l$\\
    \hline
    $v$ sends message KERNEL($idx(v)-1$,$D(v)$) over all ports.\\
    \textbf{if} $v$ is not the leaf \textbf{then}\\
    \hspace*{1em} $v$ sends message STARTKERNEL over all ports with state ``child''.\\
    \textbf{end if}\\
    \hline
    \\
    \hline
    Message KERNEL($i$,$ND$) over port $l$\\
    \hline
    Let $neigborKernel(v):= neigborKernel(v) \cup \{(l,i,ND)\}$.\\
    \textbf{if} $i=idx(v)-1$ \textbf{then}\\
    \hspace*{1em} \textbf{if} for each port $l^\prime$, there is a tuple $(l^\prime,idx(v)-1,DD)\in neighborKernel(v)$ for some $DD$ \textbf{then}\\
    \hspace*{1em}\hspace*{1em} \textbf{for each} $(l^\prime,idx(v)-1,DD) \in neigborKernel(v)$ \textbf{do}\\
    \hspace*{1em}\hspace*{1em}\hspace*{1em} $D(v):=D(v) \cap DD$. \\
    \hspace*{1em}\hspace*{1em}\hspace*{1em} $neigborKernel(v):=neigborKernel(v) \backslash
    \{(l^\prime,idx(v)-1,DD)\}$.\\
    \hspace*{1em}\hspace*{1em} \textbf{end for}\\
    \hspace*{1em}\hspace*{1em} $idx(v):=idx(v)+1$.\\
    \hspace*{1em}\hspace*{1em} \textbf{if} $idx(v) \le r$ \textbf{then}\\
    \hspace*{1em}\hspace*{1em}\hspace*{1em} $v$ sends message KERNEL($idx(v)-1$,$D(v)$) to all its neighbors.\\
    \hspace*{1em}\hspace*{1em} \textbf{else}\\
    \hspace*{1em}\hspace*{1em}\hspace*{1em} $bKernelOver(v):=true$. \\
    \hspace*{1em}\hspace*{1em}\hspace*{1em} \textbf{if} $v$ is a leaf or $v$ has
    received message KERNELOVER over all ports with state ``child''
    \textbf{then}\\
    \hspace*{1em}\hspace*{1em}\hspace*{1em}\hspace*{1em} $v$ sends message KERNELOVER over the port $l^\prime$ with state ``parent''. \\
    \hspace*{1em}\hspace*{1em}\hspace*{1em} \textbf{end if}\\
    \hspace*{1em}\hspace*{1em} \textbf{end if}\\
    \hspace*{1em} \textbf{end if}\\
    \textbf{end if}\\
    \hline
    \\
    \hline
    Message KERNELOVER over port $l$\\
    \hline
    \textbf{if} $bKernelOver(v)=true$ and $v$ has received KERNELOVER over all ports with state ``child'' \textbf{then}\\
    \hspace*{1em} \textbf{if} $v$ is not the requesting node \textbf{then}\\
    \hspace*{1em}\hspace*{1em} $v$ sends message KERNELOVER over the port with state ``parent''.\\
    \hspace*{1em} \textbf{end if}\\
    \textbf{end if}\\
    \hline
  \end{tabular*}
\end{table}

\noindent The following proposition can be proved on the $idx(v)$'s
in the above distributed algorithm.

\begin{proposition}\label{prop-kernel-idx}
During the computation of Phase III, for each node $v,w$ such that
$(v,w) \in E$, $|idx(v)-idx(w)| \le 1$.
\end{proposition}

\begin{proof}
To the contrary, suppose that $idx(v)-idx(w) > 1$ for some $v,w:
(v,w) \in E$.

From the distributed algorithm, we know that $v$ has completed the
computation of $D_{idx(v)-1}(v)$, so it has received messages
$KERNEL(idx(v)-2,DD)$ over all its ports. In particular, $v$ has
received message $KERNEL(idx(v)-2,DD)$ over the port $l^\prime$ such
that $v$ is connected to $w$ through $l^\prime$. But then, we have
$idx(w)-1 \ge idx(v)-2$, i.e. $idx(v)-idx(w) \le 1$, a
contradiction.
\end{proof}

During the computation of Phase III, for each link $(v,w) \in E$,
the number of ``KERNEL'' messages sent over $(v,w)$ is no more than
$2r$. Therefore, during the distributed computation of Phase III,
only $O(1)$ messages are sent over each link.


\section{Distributed Evaluation of FO over planar networks: Phase IV}\label{appendix-FO-planar-phase-iv}

The purpose of Phase IV is to compute $P_i:=\{ v \in V | i \in D(v),
\langle G[i,i+2r] \rangle^G \models \psi^{(r)}(v)\}$ for each $i \ge
0$.

Because our distributed algorithm for Phase IV is obtained by
transforming the centralized evaluation algorithm for MSO formulae
over classes of graphs with bounded tree-width, we first recall it
in the following.

\subsection{Centralized evaluation of MSO formulae over classes of graphs with bounded tree-width}

We first recall the centralized linear time evaluation of MSO
sentences.

Let $\Sigma$ be some alphabet. A \emph{tree language} over alphabet
$\Sigma$ is a set of rooted $\Sigma$-labeled binary trees. Let
$\varphi$ be an MSO sentence over the vocabulary $\{E_1,E_2\}\cup
\{P_c | c \in \Sigma\}$, ($E_1,E_2$ are respectively the left and
right children relations of the tree),
 the tree language accepted by
$\varphi$, $\mathcal{L}(\varphi)$, is the set of rooted
$\Sigma$-labeled trees satisfying $\varphi$.

Tree languages can also be recognized by tree automata. A
\emph{deterministic bottom-up tree automaton} $\mathcal{A}$ is a
quintuple $(Q,\Sigma, \delta, f_0,F)$, where $Q$ is the set of
states; $F \subseteq Q$ is the set of final states; $\Sigma$ is the
alphabet; and
\vspace*{-2mm}
\begin{itemize}
\item $\delta: (Q \cup Q \times Q) \times \Sigma \rightarrow Q$ is the transition function; and

\item $f_0: \Sigma \rightarrow Q$ is the initial-state assignment function.
\end{itemize}
\vspace*{-2mm}
A \emph{run} of tree automaton $\mathcal{A}\!=\!(Q,\Sigma,
\delta,f_0,F)$ over a rooted $\Sigma$-labeled binary tree
$\mathcal{T}\!=\!(T,F,r,L)$ produces a rooted $Q$-labeled tree
$\mathcal{T}^\prime\!=\!(T,F,r,L^\prime)$ such that
\vspace*{-2mm}
\begin{itemize}
\item If $t \in T$ is a leaf, then $L^\prime(t)=f_0(t)$;

\item Otherwise, if  $t \in T$ has one child $t^\prime$,
then $L^\prime(t)=\delta(L^\prime(t^\prime),L(t))$;

\item Otherwise, if $t \in T$ has two children $t_1,t_2$,
then $L^\prime(t)=\delta(L^\prime(t_1),L^\prime(t_2),L(t))$.
\end{itemize}
\vspace*{-2mm}
Note that for each deterministic bottom-up automaton $\mathcal{A}$
and rooted $\Sigma$-labeled binary tree $T$, there is exactly one
run of $\mathcal{A}$ over $T$.

The \emph{run} $\mathcal{T}^\prime=(T,F,r,L^\prime)$ of
$\mathcal{A}=(Q,\Sigma, \delta,f_0,F)$ over a rooted
$\Sigma$-labeled binary tree $\mathcal{T}=(T,F,r,L)$ is
\emph{accepting} if $L^\prime(r) \in F$.

A rooted $\Sigma$-labeled binary tree $\mathcal{T}=(T,F,r,L)$ is
{\it accepted} by a tree automaton $\mathcal{A}=(Q,\Sigma,
\delta,f_0,F)$ if the run of $\mathcal{A}$ over $\mathcal{T}$ is
accepting. The tree language {\it accepted} by $\mathcal{A}$,
$\mathcal{L}(\mathcal{A})$, is the set of rooted $\Sigma$-labeled
binary trees accepted by $\mathcal{A}$.

The next theorem shows that the two notions are equivalent.

\begin{theorem}\label{thm-MSO-automata}
\cite{TW68} Let $\Sigma$ be a finite alphabet. A tree language over
$\Sigma$ is accepted by a tree automaton iff it is defined by an MSO
sentence. Moreover, there are algorithms to construct an equivalent
tree automaton from a given MSO sentence and to construct an
equivalent MSO sentence from a given automaton.
\end{theorem}

The centralized linear time algorithm to evaluate an MSO sentence
$\varphi$ over a graph $G=(V,E)$ with tree-width bounded by $k$
works as follows:
\vspace*{-2mm}
\begin{description}
\item[Step 1] Construct an ordered tree decomposition $\mathcal{T}=(T,F,r,L)$ of $G$ of width $k$ and rank $\le 2$;

\item[Step 2] Transform $\mathcal{T}$ into a $\Sigma_k$-labeled binary tree $\mathcal{T}^\prime=(T,F,r,\lambda)$ for some finite alphabet $\Sigma_k$;

\item[Step 3] Construct an MSO sentence $\varphi^\ast$ over vocabulary
$\{E_1,E_2\} \cup \{P_c| c \in \Sigma_k\}$ from $\varphi$ (over
vocabulary $\{E\}$) such that $G \models \varphi$ iff
$\mathcal{T}^\prime \models \varphi^\ast$;

\item[Step 4] From $\varphi^\ast$, construct a bottom-up binary tree automaton $\mathcal{A}$,
and run $\mathcal{A}$ over $\mathcal{T}^\prime$ to decide whether
$\mathcal{T}^\prime$ is accepted by $\mathcal{A}$.
\end{description}
\vspace*{-2mm}
For Step 1, it has been shown that a tree decomposition of graphs
with bounded tree-width can be constructed in linear time
\cite{Bodlaender93}. It follows from Theorem~\ref{thm-MSO-automata}
that Step 4 is feasible. Now suppose that an ordered tree
decomposition $\mathcal{T}=(T,F,r,L)$ of $G=(V,E)$ of width $k$ and
rank $\le 2$ has been constructed, we recall how to perform Step 2
and Step 3 in linear time.

For Step 2, a rooted $\Sigma_{k}$-labeled tree
$\mathcal{T}^\prime=(T,F,r,\lambda)$, where $\Sigma_k=2^{[k+1]^2}
\times 2^{[k+1]^2} \times 2^{[k+1]^2}$ ($[k+1]=\{1,2,\cdots,k+1\}$),
can be obtained from $\mathcal{T}$ as follows: The new labeling
$\lambda$ over $(T,F)$ is defined by
$\lambda(t)=(\lambda_1(t),\lambda_2(t),\lambda_3(t))$, where
\vspace*{-2mm}
\begin{itemize}
\item $\lambda_1(t):=\{(j_1,j_2)\in[k+1]^2 | (b^t_{j_1},b^t_{j_2}) \in
E\}$.

\item $\lambda_2(t):=\{(j_1,j_2)\in[k+1]^2 | b^t_{j_1}=b^t_{j_2}\}$.

\item $\lambda_3(t):=
\left\{\begin{array}{cc} \{(j_1,j_2)\in[k+1]^2|
b^t_{j_1}=b^{t^\prime}_{j_2}\} &
    \mbox{for the parent } t^\prime \mbox{ of } t, \mbox{ if } t \ne r\\
\emptyset & \mbox{if } t = r
\end{array} \right.$
\end{itemize}
\vspace*{-2mm}
For Step 3, we recall how to translate the MSO sentence $\varphi$
over the vocabulary $\{E\}$ into an MSO sentence $\varphi^\ast$ over
the vocabulary $\{E_1,E_2\} \cup \{P_c| c \in \Sigma_k\}$ such that
$G \models \varphi$ iff $\mathcal{T}^\prime \models \varphi^\ast$.
The translation relies on the observation that elements  and subsets
of $V$ can be represented by  $(k+1)$-tuples of subsets of $T$. For
each element $v \in V$ and $i \in [k+1]$, let
\[
U_i(v):=\left\{\begin{array}{cl}
            \{t(v)\} & \mbox{, if }b^{t(v)}_i=v\mbox{, and }b^{t(v)}_j \ne v \mbox{ for all }j: 1 \le j < i \\
            \emptyset & \mbox{, otherwise}
          \end{array}\right.
\]
where $t(v)$ is the minimal $t \in T$ (with respect to the partial
order $\le^{\mathcal{T}}$) such that $v \in
\{b^{t}_1,\cdots,b^{t}_{k+1}\}$. Let
$\overline{U}(v)=(U_1(v),\cdots,U_{k+1}(v))$.

For each $S \subseteq V$ and $i \in [k+1]$, let $U_i(S):=\cup_{v \in
S} U_i(v)$, and let $\overline{U}(S)=(U_1(S),\cdots,U_{k+1}(S))$. It
is not hard to see that for subsets $U_1,\cdots,U_{k+1} \subseteq
T$, there exists $v \in V$ such that $\overline{U}=\overline{U}(v)$
iff
\vspace*{-2mm}
\begin{itemize}
\item (1) $\bigcup^{k+1}_{i=1} U_i$ is a singleton;

\item (2) For all $t \in T$, $i<j<k+1$, if $t \in U_j$, then $(i,j) \not \in \lambda_2(t)$;

\item (3) For all $t \in T$, $i,j<k+1$, if $t \in U_i$, then $(i,j) \not \in
\lambda_3(t)$.
\end{itemize}
\vspace*{-2mm}
\noindent Moreover, there is a subset $S \subseteq V$ such that
$\overline{U}=\overline{U}(S)$ iff conditions (2) and (3) are
satisfied. Using the above characterizations of $\overline{U}(v)$
and $\overline{U}(S)$, it is easy to construct MSO formulae
$Elem(X_1,\cdots,X_{k+1})$ and $Set(X_1,\cdots,X_{k+1})$ over
$\{E_1,E_2\} \cup \{P_c| c \in \Sigma_k\}$ such that

$\mathcal{T}^\prime \models Elem(\overline{U}) \mbox{ iff there is a
} v \in V \mbox{ such that } \overline{U}=\overline{U}(v).$

$\mathcal{T}^\prime \models Set(\overline{U}) \mbox{ iff there is a
} S \subseteq V \mbox{ such that } \overline{U}=\overline{U}(S).$

\begin{lemma}\label{lem-MSO-translation}
\cite{FFG02} Every MSO formula
$\varphi(X_1,\cdots,X_l,y_1,\cdots,y_m)$ over vocabulary $E$ can be
effectively translated into a formula
$\varphi^\ast(\overline{X}_1,\cdots,\overline{X}_l,\overline{Y}_1,\cdots,\overline{Y}_m)$
over the vocabulary $\{E_1,E_2\} \cup \{P_c | c \in \Sigma_k\}$ such
that
%
%
\begin{description}
\item (1) For all $S_1,\cdots,S_l \subseteq V$, and $v_1,\cdots,v_m \in
V$,\\
$G \models \varphi(S_1,\cdots,S_l,v_1,\cdots,v_m)$ iff
$\mathcal{T}^\prime \models
\varphi^\ast(\overline{U}(S_1),\cdots,\overline{U}(S_l),\overline{U}(v_1),\cdots,\overline{U}(v_m))$.

\item (2) For all $\overline{U}_1,\cdots,\overline{U}_l, \overline{W}_1,\cdots,\overline{W}_m \subseteq T$ such
that $\mathcal{T}^\prime \models
\varphi^\ast(\overline{U}_1,\cdots,\overline{U}_l,\overline{W}_1,$
$\cdots,\overline{W}_m)$, there exist $S_1,\cdots,S_l \subseteq V$,
$v_1,\cdots,v_m \in V$ such that $\overline{U}_i=\overline{U}(S_i)$
for all $1\le i \le l$ and $\overline{W}_j=\overline{U}(v_j)$ for
all $1 \le j \le m$.
\end{description}
%
%
\end{lemma}

\medskip

\noindent Now we recall the evaluation of MSO formulae containing
free variables over classes of graphs with bounded tree-width
\cite{FFG02}. Let $\varphi(X_1,\cdots,X_l,y_1,\cdots,y_m)$ be an MSO
formula containing free set variables $X_1,\cdots,X_l$ and
first-order variables $y_1,\cdots,y_m$.

Like the evaluation of MSO sentences, the evaluation algorithm also
consists of four steps. The first two steps of the evaluation is the
same as those of the evaluation of MSO sentences. The 3rd step is
also similar, a formula
$\varphi^\ast(\overline{X}_1,\cdots,\overline{X}_l,\overline{Y}_1,\cdots,\overline{Y}_m)$
over the vocabulary $\{E_1,E_2\} \cup \{P_c | c \in \Sigma_k\}$ is
obtained from $\varphi(X_1,\cdots,X_l,y_1,\cdots,y_m)$ (over the
vocabulary $\{E\}$) such that the conditions specified in
Lemma~\ref{lem-MSO-translation} are satisfied. The main difference
is in the 4th step.

Because $\varphi^\ast$ is not a sentence and
Theorem~\ref{thm-MSO-automata} only applies to MSO sentences, we
cannot construct directly a tree automaton from $\varphi^\ast$ and
run the automaton over $\mathcal{T}^\prime$. However, we can replace
the free set variables in
$\varphi^\ast(\overline{X}_1,\cdots,\overline{X}_l,\overline{Y}_1,\cdots,\overline{Y}_m)$
by some appropriate new unary relation names and transform it into a
sentence $\varphi^{\ast\ast}$. Let $\Sigma^\prime_k:=\Sigma_k \times
\{0,1\}^{(k+1)(l+m)}$, then from $\varphi^{\ast\ast}$, an automaton
$\mathcal{A}=(Q,\Sigma^\prime_k, \delta, f_0, F)$ can be constructed
such that for each $\Sigma^\prime_k$-labeled tree
$\mathcal{S}^\prime$, $\mathcal{S}^\prime \models
\varphi^{\ast\ast}$ if and only if $\mathcal{A}$ accepts
$\mathcal{S}^\prime$.

A $\Sigma_k$-labeled tree $\mathcal{S}=(S,H,r,M)$ together with
$\overline{U}_1,\cdots,\overline{U}_l,
\overline{W}_1,\cdots,\overline{W}_m \subseteq S$ leads to a
$\Sigma^\prime_k$-labeled tree $(S,H,r,M^\prime)$, denoted by
$(\mathcal{S};\overline{U}_1,\cdots,\overline{U}_l,
\overline{W}_1,\cdots,\overline{W}_m)$, in a natural way:
$M^\prime(s)=(M(s),\bar{\varepsilon},\bar{\theta})$, where
\[\varepsilon_{(k+1)(i-1)+j}=1  \mbox{ iff } s \in U^j_i \mbox{ for all } 1 \le i \le
l, 1 \le j \le k+1, \mbox{ and}\]
\[\theta_{(k+1)(i-1)+j}=1 \mbox{ iff } s \in
W^j_i \mbox{ for all } 1 \le i \le m, 1 \le j \le k+1.\]

Then given a $\Sigma_k$-labeled tree $\mathcal{S}$, the computation
of the set
\[\varphi^\ast(\mathcal{S}):=\left\{\overline{U}_1,\cdots,\overline{U}_l,
\overline{W}_1,\cdots,\overline{W}_m \subseteq S | \mathcal{S}
\models \varphi^\ast(\overline{U}_1,\cdots,\overline{U}_l,
\overline{W}_1,\cdots,\overline{W}_m)\right\}\] can be reduced to
the computation of the set
\[\mathcal{A}(\mathcal{S}):=\left\{\overline{U}_1,\cdots,\overline{U}_l,
\overline{W}_1,\cdots,\overline{W}_m \subseteq S | \mathcal{A}
\mbox{ accepts } (\mathcal{S};\overline{U}_1,\cdots,\overline{U}_l,
\overline{W}_1,\cdots,\overline{W}_m) \right\}.\]

Now we recall how $\mathcal{S}=(S,H,r,M)$ can be passed by
$\mathcal{A}=(Q,\Sigma^\prime_k,\delta, f_0,F)$ for three times,
first in bottom-up, then top-down, finally bottom-up again, to
compute $\mathcal{A}(\mathcal{S})$.

(1) \emph{Bottom-up}. From leaves to the root, for each $s \in S$,
the set of ``potential states'' of $s$, denoted $Pot_s$, is computed
inductively: If $s$ is a leaf, then
$Pot_s:=\{f_0(M(s),\bar{\varepsilon},\bar{\theta})|
\bar{\varepsilon} \in \{0,1\}^{l(k+1)}, \bar{\theta} \in
\{0,1\}^{m(k+1)}\}$. For an inner vertex $s$ with a child
$s^\prime$,
\[Pot_s:=\{\delta(q^\prime,(M(s),\bar{\varepsilon},\bar{\theta}))|
q^\prime \in Pot_{s^\prime}, \bar{\varepsilon} \in \{0,1\}^{l(k+1)},
\bar{\theta} \in \{0,1\}^{m(k+1)}\}.\]

For an inner vertex $s$ with two children $s_1$ and $s_2$,
\[Pot_s:=\{\delta(q_1,q_2,(M(s),\bar{\varepsilon},\bar{\theta}))| q_1 \in
Pot_{s_1}, q_2 \in Pot_{s_2}, \bar{\varepsilon} \in
\{0,1\}^{l(k+1)}, \bar{\theta} \in \{0,1\}^{m(k+1)}\}.\]

(2) \emph{Top-down}. Starting from the root $r$, for each $s \in S$,
the set of ``successful states'' of $s$, denoted $Suc_s$, is
computed: let $Suc_r:=F \cap Pot_r$, and for $s \in S$ with parent
$t$ and no sibling,
\[Suc_s:=\{q \in Pot_s | \exists \bar{\varepsilon},\bar{\theta}, \mbox{ such that } \delta(q,(M(t),\bar{\varepsilon},\bar{\theta})) \in Suc_{t}\}.\]

For $s \in S$ with parent $t$ and a sibling $s^\prime$,
\[Suc_s:=\{q \in Pot_s | \exists q^\prime \in Pot_{s^\prime}, \bar{\varepsilon},\bar{\theta}, \mbox{ such that }\delta(q,q^\prime, (M(t),\bar{\varepsilon},\bar{\theta})) \in Suc_{t}\}.\]

(3) \emph{Bottom-up again}. For $s \in S$, let $\mathcal{S}_s$
denote the subtree of $\mathcal{S}$ with $s$ as the root. Starting
from the leaves, for each $s \in S$ and $q \in Suc_s$, compute
$Sat_{s,q}$. Intuitively, a tuple $\bar{B},\bar{C} \subseteq S_s$ is
in $Sat_{s,q}$ if it is the restriction of a ``satisfying
assignment'' $\overline{B^\prime},\overline{C^\prime} \in
\mathcal{A}(\mathcal{S})$ to $S_s$, and for the run of $\mathcal{A}$
over $(\mathcal{S};\overline{B^\prime},\overline{C^\prime})$, the
state of the run at $s$ is $q$.

Let $s \in S$ and $q \in Suc_s$. Set $B^s_1:=\{s\}$ and
$B^s_0:=\emptyset$.

If $s$ is a leaf, then
\[Sat_{s,q}:=\left\{\left.(B_{\varepsilon_{1}},\cdots,B_{\varepsilon_{l(k+1)}},C_{\theta_{1}},\cdots,C_{\theta_{m(k+1)}})\right| q=f_0(M(s),\bar{\varepsilon},\bar{\theta})\right\}.\]

If $s$ is an inner vertex with one child $s^\prime$, then
\[
Sat_{s,q}:=\left\{\begin{array}{l} \left.(B^\prime_1 \cup
B_{\varepsilon_{1}},\cdots,B^\prime_{l(k+1)} \cup
B_{\varepsilon_{l(k+1)}},C^\prime_1 \cup
C_{\theta_{1}},\cdots,C^\prime_{m(k+1)} \cup C_{\theta_{m(k+1)}})
\right| \\
\exists q^\prime \in Suc_{s^\prime} \mbox{ such that, }
q=\delta(q^\prime,(M(s),\bar{\varepsilon},\bar{\theta})),
(\overline{B^\prime},\overline{C^\prime}) \in
Sat_{s^\prime,q^\prime}.
\end{array}\right\}.\]

If $s$ is an inner vertex with two children $s_1$ and $s_2$, then
\[
Sat_{s,q}:=\left\{\begin{array}{l} \left.\left(\begin{array}{l}
B^\prime_1 \cup B^{\prime\prime}_1 \cup
B_{\varepsilon_{1}},\cdots,B^\prime_{l(k+1)} \cup
B^{\prime\prime}_{l(k+1)} \cup B_{\varepsilon_{l(k+1)}},\\
C^\prime_1 \cup C^{\prime\prime}_1 \cup
C_{\theta_{1}},\cdots,C^\prime_{m(k+1)} \cup
C^{\prime\prime}_{m(k+1)} \cup C_{\theta_{m(k+1)}}
\end{array}\right) \right| \\
\begin{array}{l}
\exists q_1 \in Suc_{s_1}, q_2 \in Suc_{s_2}\mbox{ such that, } q=\delta(q_1,q_2,(M(s),\bar{\varepsilon},\bar{\theta})),\\
(\overline{B^\prime},\overline{C^\prime}) \in Sat_{s_1,q_1},
(\overline{B^{\prime\prime}},\overline{C^{\prime\prime}}) \in
Sat_{s_2,q_2}.
\end{array}\end{array}\right\}.\]

Then $\mathcal{A}(\mathcal{S})=\bigcup_{q \in Suc_r} Sat_{r,q}$.

Therefore, we can run $\mathcal{A}$ over $\mathcal{T}^\prime$ for
three times to compute $\mathcal{A}(\mathcal{T}^\prime)$. Finally
from $\mathcal{A}(\mathcal{T}^\prime)$, we can construct
$\varphi(G)=\left\{\left.(S_1,\cdots,S_l,v_1,\cdots,v_m) \right| G
\models \varphi(S_1,\cdots,S_l,v_1,\cdots,v_m)\right\}$ according to
the mechanism to encode the elements and sets of $V$ into the
subsets of $\mathcal{T}^\prime$.

\subsection{Distributed evaluation of $\psi^{(r)}(x)$ over $G[i,i+2r]$'s}

Now we consider the distributed evaluation of $\psi^{(r)}(x)$ over
$G[i,i+2r]$'s.

Because $\psi^{(r)}(x)$ is a local formula, it is sufficient to
evaluate $\psi^{(r)}(x)$ over each connected component of
$G[i,i+2r]$.

Let $C_i$ be a connected component of $G[i,i+2r]$, and
$w^i_1,\cdots,w^i_l$ be the nodes contained in $C_i$ with distance
$i$ from the requesting node. Now we consider the evaluation of
$\psi^{(r)}(x)$ over $C_i$.

Let $C^\prime_i$ be the graph obtained from $C_i$ by including all
ancestors of $w^i_1,\cdots,w^i_l$, and $C^\ast_i$ be the graph
obtained from $C^\prime_i$ by contracting all the ancestors of
$w^i_1,\cdots,w^i_l$ into one vertex $v^\ast$, i.e. $C^\ast_i$ has
one more vertex $v^\ast$ than $C_i$, and $v^\ast$ is connected to
$w^i_1,\cdots,w^i_l$. It is easy to see that $C^\ast_i$ is a planar
graph with a BFS tree rooted on $v^\ast$ with depth at most $2r+1$.
Consequently $C^\ast_i$ is a planar graph with bounded diameter,
thus a graph with bounded tree-width. Because $C_i$ is a subgraph of
$C^\ast_i$, $C_i$ is a planar graph with bounded tree-width as well.

Our purpose is to construct distributively an ordered tree
decomposition for $C_i$, and evaluate $\psi^{(r)}(x)$ by using the
automata-theoretic technique.


The distributed construction of an ordered tree decomposition for a
planar network with bounded diameter is as follows \cite{GW09}:
\begin{itemize}
\item Do a depth-first-search to decompose the network into blocks,
i.e. biconnected components;

\item Construct an ordered tree decomposition for each
nontrivial block: Traverse every face of the block according to the
cyclic ordering at each node, triangulate all those faces, and
connect the triangles into a tree decomposition by utilizing the
pre-computed BFS tree;

\item Finally the tree decompositions for the blocks are connected
together into a complete tree decomposition for the whole network.
\end{itemize}

The blocks of $C^\ast_i$ enjoy the following property.

\begin{lemma}\label{lem-FO-planar-T-biconn}
Let
\begin{itemize}
\item $C_i$ be a connected component of $G[i,i+2r]$,

\item $w^i_1,\cdots,w^i_l$ be all the nodes contained in $C_i$ with
distance $i$ from the requesting node,

\item $C^\prime_i$ be the graph obtained from $C_i$ by including all
ancestors of $w^i_1,\cdots,w^i_l$,

\item $C^\ast_i$ be the graph obtained from $C^\prime_i$ by contracting all the ancestors of
$w^i_1,\cdots,w^i_l$ into one vertex.
\end{itemize}

Then the virtual vertex $v^\ast$ and all the $w^i_1,\cdots,w^i_l$
are contained in a unique block $B_0$ of $C^{\ast}_i$, and for each
block $B \ne B_0$, there is a $w^i_j$ such that \[V(B) \subseteq \{u
\in V(C_i) | u \mbox{ is a descendant of } w^i_j \mbox{ in the BFS
tree}\}.\]
\end{lemma}

The distributed tree decomposition of $C_i$ can be constructed as
follows: Starting from some $w^i_j$ ($1 \le j \le l$), do a
depth-first-search to decompose $C^\ast_i$ into blocks by imagining
that there is a virtual node $v^\ast$, then $v^\ast$ and all
$w^i_1,\cdots,w^i_l$ belong to a unique biconnected component $B_0$.
Construct an ordered tree decomposition for each block, and do some
special treatments for $B_0$ (when the virtual node $v^\ast$ is
visited). Finally connect these tree decompositions together in a
suitable way to get a complete tree decomposition of $C_i$.

Moreover, a post-order traversal over the BFS tree can be done to
construct the tree decompositions for connected components of all
$G[i,i+2r]$'s one by one.

With the ordered tree decomposition for $C_i$, $\psi^{(r)}(x)$ can
be evaluated over $C_i$ as follows: the node $w^i_j$ first
transforms $\psi^{(r)}(x)$ into a formula
$\psi^\ast(U_1,\cdots,U_{k+1})$ over the vocabulary $\{E_1,E_2\}
\cup \left\{P_c | c \in \Sigma_k\right\}$ satisfying the condition
in Lemma~\ref{lem-MSO-translation}. Then from $\psi^\ast$,
constructs an automaton $\mathcal{A}$ over $\Sigma^\prime_k$-labeled
trees, and sends $\mathcal{A}$ to all the nodes in $C_i$. The
ordered tree decomposition is then transformed into a
$\Sigma_k$-labeled tree $\mathcal{T}^\prime$. Finally $\mathcal{A}$
is ran over $\mathcal{T}^\prime$ for three times to get
$\mathcal{A}(\mathcal{T}^\prime)$, and the evaluation result of
$\psi^{(r)}(x)$ over $C_i$ is distributively stored on the nodes of
$C_i$.

Because the most intricate part of Phase IV is the distributed
construction of an ordered tree decomposition for each connected
component $C_i$ of $G[i,i+2r]$. In the following, we only illustrate
how to do a post-order traversal of the BFS tree to decompose each
connected component $C_i$ of $G[i,i+2r]$ into blocks and construct
an ordered tree decomposition for each block of $C_i$, and omit the
other parts of Phase IV.

\begin{table}[ht]
  \centering
  \small
\begin{tabular*}{\textwidth}{l@{\extracolsep{\fill}}}
    \hline
    Initialization\\
    \hline
    The requesting node sets $traversed(1):=true$, and sends message POSTTRAVERSE over port $1$.\\
    \hline
    \\
    \hline
    Message POSTTRAVERSE over port $l$\\
    \hline
    \% \textsf{Without loss of generality, suppose that $treeDepth > 2r$.}\\
    \textbf{if} $depth(v) + 2r = treeDepth$ or $v$ is a leaf \textbf{then}\\
    \hspace*{1em} \textbf{if} $treeDecompOver(depth(v))=false$ \textbf{then}\\
    \hspace*{1em}\hspace*{1em} $DFSDepth(depth(v),v):=1$, $DFSLow(depth(v),v):=0$.\\
\hline
  \end{tabular*}
\end{table}

\begin{table}[ht]
  \centering
  \small
  \begin{tabular*}{\textwidth}{l@{\extracolsep{\fill}}}
    \hline
    Message DFSFORWARD($rBFSDepth$, $rId$, $nextBlockId$, $parentDFSDepth$) over port $l$\\
    \hline
    \textbf{if} $DFSVisited(rBFSDepth, v)=false$ \textbf{then}\\
    \hspace*{1em} $DFSVisited(rBFSDepth, v):=true$, $DFSState(rBFSDepth,l):=$``parent''.\\
    \hspace*{1em} $DFSRootId(rBFSDepth):=rId$, $DFSDepth(rBFSDepth, v):=parentDFSDepth+1$.\\
    \hspace*{1em} $DFSLow(rBFSDepth,v):=DFSDepth(rBFSDepth, v)$.\\
    \hspace*{1em} \textbf{if} $depth(v)=rBFSDepth$ \textbf{then}\\
    \hspace*{1em}\hspace*{1em} \textbf{for each} port $l^\prime \ne l$ such that $state(l^\prime)=$``child'' or ``horizon'' or ``downward'' \textbf{do}\\
    \hspace*{1em}\hspace*{1em}\hspace*{1em} $DFSState(rBFSDepth,l^\prime):=$``unvisited''.\\
\hline
  \end{tabular*}
\end{table}

\begin{table}[ht]
  \centering
  \small
  \begin{tabular*}{\textwidth}{l@{\extracolsep{\fill}}}
    \hline
    Message DFSBACKTRACK($rBFSDepth$, $nextBlockId$, $childDFSLow$) over port $l$\\
    \hline
    \textbf{if} $childDFSLow = DFSDepth(rBFSDepth,v)$ \textbf{then}\\
    \hspace*{1em} $DFSState(rBFSDepth,l):=$``closed''.\\
    \hspace*{1em} $blockIds(rBFSDepth):=blockIds(rBFSDepth) \cup \{nextBlockId\}$.\\
    \hspace*{1em} $blockPorts(rBFSDepth,nextBlockId):=\{l\}$.\\
    \hspace*{1em} $v$ sends message DFSINFORM$(rBFSDepth,nextBlockId)$ over $l$.\\
    \hspace*{1em} $nextBlockId:=nextBlockId+1$.\\
    \textbf{else if} $childDFSLow > DFSDepth(rBFSDepth,v)$ \textbf{then} \\
\hline
  \end{tabular*}
\end{table}

\begin{table}[ht]
  \centering
  \small
  \begin{tabular*}{\textwidth}{l@{\extracolsep{\fill}}}
    \hline
    Message DFSINFORM($rBFSDepth$, $blockId$) over port $l$.\\
    \hline
    \textbf{if} $blockId \not \in blockIds(rBFSDepth)$ \textbf{then}\\
    \hspace*{1em} $blockIds(rBFSDepth):=blockIds(rBFSDepth) \cup \{blockId\}$.\\
    \hspace*{1em} $blockPorts(rBFSDepth,blockId):=$\\
    \hspace*{1em}\hspace*{1em} $\{l^\prime | DFSState(rBFSDepth,l^\prime)=\mbox{``parent'' or ``non-tree-backward'' or ``backtracked''}\}$.\\
    \hspace*{1em} \textbf{if} there are ports $l^\prime$ such that $DFSState(rBFSDepth,l^\prime)=$``non-tree-backward'' \textbf{then}\\
    \hspace*{1em}\hspace*{1em} $v$ sends DFSBLOCKPORT$(rBFSDepth,blockId)$ \\
    \hspace*{1em}\hspace*{1em}\hspace*{1em} over all ports $l^\prime$ such that $DFSState(rBFSDepth,l^\prime)=$``non-tree-backward''.\\
\hline
  \end{tabular*}
\end{table}

\begin{table}[ht]
  \centering
  \small
  \begin{tabular*}{\textwidth}{l@{\extracolsep{\fill}}}
    \hline
    Message DFSBLOCKACK($rBFSDepth$) over port $l$.\\
    \hline
    $bDFSBlockAck(rBFSDepth,l):=true$.\\
    \textbf{if} $v = DFSRootId(rBFSDepth)$ \textbf{then}\\
    \hspace*{1em} \textbf{if} $bDFSBlockOver(rBFSDepth,l^\prime)=true$ for each $l^\prime$ such that \\
    \hspace*{1em}\hspace*{1em} $DFSState(rBFSDepth,l^\prime)=$``backtracked'', \\
    \hspace*{1em}\hspace*{1em} and $bDFSBlockAck(rBFSDepth,l^\prime)=true$ for each port $l^\prime$\\
    \hspace*{1em}\hspace*{1em} such that $DFSState(rBFSDepth,l^\prime)=$``closed'' or ``backtracked'' or ``childBridge'' \textbf{then}\\
    \hspace*{1em}\hspace*{1em} $v$ sends messages to construct a tree decomposition for each block.\\
\hline
  \end{tabular*}
\end{table}

\begin{table}[ht]
  \centering
  \small
  \begin{tabular*}{\textwidth}{l@{\extracolsep{\fill}}}
    \hline
    Message DFSRESTART$(rBFSDepth, nextBlockId, ancestorDFSDepth)$ over port $l$.\\
    \hline
    $DFSState(rBFSDepth,l)$:=``non-tree-backward''.\\
    $DFSLow(rBFSDepth,v):=\min\{DFSLow(rBFSDepth,v),ancestorDFSDepth\}$. \\
    \textbf{if} $depth(v)=rBFSDepth$ \textbf{then}\\
    \hspace*{1em} \textbf{if} there exist at least one port $l^\prime$ such that $DFSState(rBFSDepth,l^\prime):=$``unvisited'' \\
    \hspace*{1em}\hspace*{1em} and $state(l^\prime):=$``child'' or ``horizon'' or ``downward'' \textbf{then}\\
    \hspace*{1em}\hspace*{1em} Let $l^\prime$ be the minimal such port, $DFSState(rBFSDepth, l^\prime):=$``child''.\\
    \hspace*{1em}\hspace*{1em} $v$ sends DFSFORWARD($rBFSDepth$, $DFSRootId(rBFSDepth)$, \\
\hline
  \end{tabular*}
\end{table}

\clearpage

Let $C_i$ be a connected component of $G[i,i+2r]$, and
$w^i_1,\cdots,w^i_l$ be all the nodes contained in $C_i$ with
distance $i$ from the requesting node. In the following, we will
consider the distributed construction of an ordered tree
decomposition for the special block of $C_i$ which contains all
$w^i_1,\cdots,w^i_l$, by imagining that there is a virtual vertex
$v^\ast_i$ connected to all $w^i_1,\cdots,w^i_l$, and illustrate the
special treatments that should be done.

\begin{table}[ht]
  \centering
  \small
  \begin{tabular*}{\textwidth}{l@{\extracolsep{\fill}}}
    \hline
    \% \textsf{The distributed construction of an ordered tree
decomposition for the special block of $C_i$, }\\
    \hspace*{1em} \textsf{by imagining that there is a virtual vertex and doing some special treatments.}\\
    \hline
    \\
    \hline
    Initialization\\
    \hline
    Let $u_0$ satisfy that $DFSRoot(depth(u_0))=u_0$.\\
    Suppose for each node $v$ and $i \in C(v)$, $v$ stores in $BFSAncestors(i)$ \\
    \hspace*{1em} a list of all its ancestors that are of depth from $i$ to $depth(v)$ in the BFS tree.\\
    Suppose for each $v$ and port $l^\prime$, $v$ stores in $neighbor(l^\prime)$ the
    neighbor of $v$ corresponding to $l^\prime$.\\
\hline
  \end{tabular*}
\end{table}

\begin{table}[ht]
  \centering
  \small
  \begin{tabular*}{\textwidth}{l@{\extracolsep{\fill}}}
    \hline
    Message DFSFACESTART(($rBFSDepth$,$spBlockId$, $u_1$, $u_2$), $(w_1,\cdots,w_r)$, $(w^\prime_1,\cdots,w^\prime_s)$) over port $l$\\
    \hline
    $arcVisited(rBFSDepth,spBlockId,neighbor(l),v):=true$.\\
    Let $l^\prime$ be the port such that $(v,neighbor(l^\prime))$ is immediately before $(v,neighbor(l))$ in the cyclic
    ordering.\\
    \textbf{if} $depth(v)=rBFSDepth$ and $state(l^\prime)=$``parent'' or ``upward'' \textbf{then}\\
    \hspace*{1em} $v$ sends DFSSPECIALTREAT($rBFSDepth$,$spBlockId$,$u_1$,$u_2$,$v$,$BFSAncestors(rBFSDepth)$) over $l$.\\
    \textbf{else if} $depth(v)=rBFSDepth+2r$ and $state(l^\prime)=$``child'' or ``downward'' \textbf{then}\\
    \hspace*{1em} Let $l^{\prime\prime}$ be the port satisfying that $(v,neighbor(l^{\prime\prime}))$ is the first arc before $(v,neighbor(l))$\\
    \hspace*{1em}\hspace*{1em} in the cyclic ordering such that $l^{\prime\prime} \in blockPorts(rBFSDepth,spBlockId)$.\\
\hline
  \end{tabular*}
\end{table}

\begin{table}[ht]
  \centering
  \small
  \begin{tabular*}{\textwidth}{l@{\extracolsep{\fill}}}
    \hline
    Message DFSACKFACESTART(($rBFSDepth$,$spBlockId$,$u_1$,$u_2$), $(w_1,\cdots,w_r)$, $(w^\prime_1,\cdots,w^\prime_s)$) over port $l$\\
    \hline
    $Bag(rBFSDepth,spBlockId,u_1,v,neighbor(l)):=$\\
    \hspace*{1em} $\langle (u_1,u_2),list_{3k+1}((w_1,\cdots,w_r) \cdot BFSAncestors(rBFSDepth) \cdot (w^\prime_1,\cdots,w^\prime_s))\rangle$.\\
    \hline
    \\
    \hline
    Message DFSFACEOVER(($rBFSDepth$,$spBlockId$,$u_1$,$u_2$), $v^\prime$, $(w_1,\cdots,w_r)$, $(w^\prime_1,\cdots,w^\prime_s)$) over port $l$\\
    \hline
    $arcVisited(rBFSDepth,spBlockId,neighbor(l),v):=true$.\\
    $Bag(rBFSDepth,spBlockId,v,v^\prime, neighbor(l)):=$\\
    \hspace*{1em} $\langle (u_1,u_2),list_{3k+1}(BFSAncestors(rBFSDepth) \cdot (w_1,\cdots,w_r) \cdot (w^\prime_1,\cdots,w^\prime_s))\rangle$.\\
\hline
    \end{tabular*}
\end{table}

\begin{table}[ht]
  \centering
  \small
  \begin{tabular*}{\textwidth}{l@{\extracolsep{\fill}}}
    \hline
    Message DFSPOSTBACKTRACK($rBFSDepth,spBlockId$) over port $l$\\
    \hline
    \textbf{if} there exists $l^\prime$ such that $DFSState(rBFSDepth,l^\prime)=$``backtracked'' and\\
    \hspace*{1em} $DFSPostTraversed(rBFSDepth,l^\prime)=false$ \textbf{then}\\
    \hspace*{1em} Let $l^\prime$ be the minimal port such that $DFSState(rBFSDepth,l^\prime)=$``backtracked''\\
    \hspace*{1em}\hspace*{1em} and $DFSPostTraversed(rBFSDepth,l^\prime)=false$.\\
    \hspace*{1em} $DFSPostTraversed(rBFSDepth,l^\prime):=true$.\\
    \hspace*{1em} $v$ sends DFSPOSTTRAVERSE($rBFSDepth$,$spBlockId$) over $l^\prime$.\\
    \textbf{else if} there exists $l^\prime$ such that $l^\prime \in blockPorts(rBFSDepth,spBlockId)$ and \\
\hline
    \end{tabular*}
\end{table}

\begin{table}[ht]
  \centering
  \small
  \begin{tabular*}{\textwidth}{l@{\extracolsep{\fill}}}
  \hline
    Message DFSSPECIALFACESTART($u_0$, ($rBFSDepth$,$spBlockId$,$u_1$,$u_2$),$(w_1,\cdots,w_r)$, $(w^\prime_1,\cdots,w^\prime_s)$) over $l$\\
    \hline
    $arcVisited(rBFSDepth,spBlockId,neighbor(l),v):=true$.\\
    Let $l^\prime$ be the port such that $(v,neighbor(l^\prime))$ is immediately before $(v,neighbor(l))$ in the cyclic
    ordering.\\
    \textbf{if} $depth(v)=rBFSDepth$ and $state(l^\prime)=$``parent'' or ``upward'' \textbf{then}\\
    \hspace*{1em} \textbf{if} $v=u_2$ \textbf{then}\\
    \hspace*{1em}\hspace*{1em} $Bag(rBFSDepth,spBlockId,u_1,v,v):=$\\
    \hspace*{1em}\hspace*{1em}\hspace*{1em} $\langle (u_1,u_2),list_{3k+1}((w_1,\cdots,w_r) \cdot BFSAncestors(rBFSDepth)) \cdot
    BFSAncestors(rBFSDepth))\rangle$.\\
    \hspace*{1em} \textbf{else}\\
\hline
    \end{tabular*}
\end{table}

\begin{table}[ht]
  \centering
  \small
  \begin{tabular*}{\textwidth}{l@{\extracolsep{\fill}}}
    \hline
    Message DFSSPECIALACKFACESTART(($rBFSDepth$, $spBlockId$, $u_1$, $u_2$),\\
    \hspace*{1em} $(w_1,\cdots,w_r)$, $(w^\prime_1,\cdots,w^\prime_s)$) over $l$\\
    \hline
    $Bag(rBFSDepth,spBlockId,u_1,v,neighbor(l)):=$\\
    \hspace*{1em} $\langle (u_1,u_2),list_{3k+1}((w_1,\cdots,w_r) \cdot BFSAncestors(rBFSDepth) \cdot (w^\prime_1,\cdots,w^\prime_s))\rangle$.\\
    \hline
    \\
    \hline
    Message DFSSPECIALFACEWALK($u_0$, ($rBFSDepth$,$spBlockId$,$u_1$,$u_2$),$(w_1,\cdots,w_r)$, $(w^\prime_1,\cdots,w^\prime_s)$) over $l$\\
    \hline
    $arcVisited(rBFSDepth,spBlockId,neighbor(l),v):=true$.\\
    $Bag(rBFSDepth,spBlockId,u_1,neighbor(l),v):=$\\
\hline
    \end{tabular*}
\end{table}
\clearpage

\section{Distributed Evaluation of FO over planar networks: Phase V}\label{appendix-FO-planar-phase-v}

Label nodes in $\bigcup_i P_i$ with $P$.

Then consider the evaluation of FO sentence $\varphi^\prime$ over
the vocabulary $\{E,P\}$, where

\begin{equation*}
\varphi^\prime:=\exists x_1 ... \exists x_s \left(\bigwedge
\limits_{1 \le i<j \le s} d(x_i,x_j)
> 2r \wedge \bigwedge \limits_{i} P(x_i)\right).
\end{equation*}

Starting from some node $w_1$ with label $P$, mark the vertices in
$N_{2r}(w_1)$ as $Q$, then select some node $w_2$ outside $Q$, and
mark those nodes in $N_{2r}(w_2)$ by $Q$ again, continue like this,
until $w_l$ such that either $l = s$ or all the nodes with label $P$
have already been labeled by $Q$.

If $l < s$, then label the nodes in $\bigcup \limits_{1 \le i \le l}
N_{4r}(v_i)$ as $I$. Then each connected component of $\langle I
\rangle^G$ has diameter no more than $4lr < 4sr$. We can construct
distributively a tree decomposition for each connected component of
$\langle I \rangle^G$, and connect these tree decompositions
together to get a complete tree-decomposition of $\langle I
\rangle^G$, then evaluate the sentence $\varphi^\prime$ by using
this complete tree decomposition.

\section{The proof of Theorem~\ref{thm:FO-UCNT-frugal}}

\noindent The proof of Theorem~\ref{thm:FO-UCNT-frugal} relies on a
normal form of FO($\#$) formulae.

\begin{lemma}\label{lem:FO-UCNT-norm}
FO($\#$) formulae can be rewritten into a Boolean combinations of
(i) first-order formulae and (ii) sentences of the form $t_1 =t_2$
or $t_1 < t_2$ where $t_i$ are second sort terms, and for each
second sort term $\# x. \varphi(x)$ occurring in $t_i$, $\varphi(x)$
is a first-order formula.
\end{lemma}

The proof of the lemma can be done by a simple induction on the
syntax of FO($\#$) formulae.

\begin{proof} Theorem~\ref{thm:FO-UCNT-frugal} (sketch)

From Theorem~\ref{thm:FO-bounded-degree} and
Theorem~\ref{thm:FO-planar}, we know that FO formulae can be
evaluated over bounded degree and planar networks with only a
bounded number of messages sent over each link. From the normal form
of FO($\#$) formulae (Lemma~\ref{lem:FO-UCNT-norm}), it is
sufficient to prove that sentences of the form $t_1 = t_2$ or $t_1 <
t_2$ can be frugally evaluated over the two types of networks.

By induction, we can show that for all second sort terms $t$, $t^G$
is bounded by $n^{|t|}$ (where $|t|$ is the number of symbols in
$t$, and $n$ is the size of $V$). Therefore, $t^G$ can be encoded in
$O(\log n)$ bits.

At first we consider the computation of the term $\#x. \varphi(x)$
($\varphi$ is a first-order formula with only one free variable
$x$).

The requesting node starts the frugal evaluation of $\varphi(x)$
(Theorem~\ref{thm:FO-Unary}), then each node $v$ knows whether
$\varphi(v)$ holds or not. Now the requesting node can aggregate the
result of $\#x.\varphi(x)$ by using the pre-computed BFS-tree.

If $t_1$ and $t_2$ can be frugally computed, then $t_1 + t_2$, $t_1
- t_2$ and $t_1 \times t_2$ can be frugally computed as well by just
computing $t_1$ and $t_2$ separately, and computing $t_1 + t_2$,
$t_1 - t_2$ or $t_1 \times t_2$ by in-node computation. Thus all
FO($\#$) sentences of the form $t_1 = t_2$ and $t_1 < t_2$ can be
frugally computed.
\end{proof}

\end{appendix}

\end{document}